\documentclass[12pt]{article}
\usepackage{enumerate}
\usepackage{amsfonts}
\usepackage{latexsym}
\usepackage{color}
\usepackage{graphicx}
\usepackage{wrapfig}

\usepackage{amsmath,amssymb,amsthm}

\theoremstyle{plain}
\newtheorem{theorem}{Theorem}
\newtheorem{lemma}[theorem]{Lemma}
\newtheorem{proposition}[theorem]{Proposition}
\newtheorem{corollary}[theorem]{Corollary}

\theoremstyle{definition}

\newtheorem{remark}[theorem]{Remark}

\newcommand{\HH}{\mathcal H}
\newcommand{\St}{\mathfrak S}
\newcommand{\Tr}{\operatorname{Tr}}
\newcommand{\B}{\mathfrak{B}}
\newcommand{\T}{\mathfrak{T}}

\newcommand{\Bsa}{\mathfrak{B}_{\mathrm{sa}}}
\newcommand{\supp}{\operatorname{supp}}

\newcommand{\rank}{\operatorname{rank}}

\newcommand{\shs}{\hspace{1pt}}

\newcommand{\ext}{\mathrm{ext}}
\begin{document}
\title{Log-Sobolev inequality, von Neumann entropy and Entanglement of Formation}
\author{A.S.Holevo, M.E.~Shirokov\\
Steklov Mathematical Institute, Moscow, Russia}
\date{}
\maketitle

\maketitle

\begin{abstract}  We present two results derived from the sharp log-Sobolev inequality for the uniform measure on a complete graph which concern
the von Neumann entropy and the Entanglement of Formation of a state of finite and infinite-dimensional quantum systems.

The first result is a sharp Lipschitz lower semicontinuity bound for the von Neumann entropy at any mixed state $\rho$ with uniform positive spectrum (i.e. a state proportional to a projector) w.r.t. the fidelity deficit:  the inequality   $\,S(\rho)-S(\sigma)\leq C_\rho(1-F(\rho,\sigma))\,$ valid for any state $\sigma$, where $C_{\rho}$ is a constant depending on the rank of $\rho$.

The second result is a sharp  Lipschitz  lower semicontinuity bound for the Entanglement of Formation at any pure state $\rho$ with uniform positive spectrum of marginal states w.r.t. the fidelity deficit: the inequality
$\,E_F(\rho)-E_F(\sigma)\leq C_\rho(1-\mathrm{Tr}\rho\sigma)\,$ valid for any state $\sigma$, where $C_{\rho}$ is a constant depending on the Schmidt rank of $\rho$.

In both cases the optimal constant $C_\rho$ is equal to the optimal constant $K_{d}$ in the log-Sobolev inequality for the complete graph with $d$ vertices: in the
first case $d=\rank \rho$,  in the second one $d=\rank \rho_A=\rank \rho_B$.

The authors are grateful to  GPT 5.6 for valuable discussion and technical help in preparing this note.
\end{abstract}

\tableofcontents

\section{Introduction}

In this article we show how the log-Sobolev inequality for the uniform measure on a complete graph (obtained in \cite{LS-1} and described in Section 2)
can be used to obtain sharp
lower semicontinuity bounds for two important characteristics of quantum
systems: the von Neumann entropy and the Entanglement of Formation.

The term "lower semicontinuity bound for a function $f$ at a state $\rho$" means the
inequality
\[
  f(\rho) - f(\sigma) \le B_{\rho}\bigl(d(\rho,\sigma)\bigr)
\]
valid for this state $\rho$ and  any other state $\sigma$ of a quantum system, where $B_{\rho}$ is the
function depending on $\rho$ (and not depending on $\sigma$) and $d(\rho,\sigma)$
is a measure of divergence between $\rho$ and $\sigma$ (the trace norm distance
$\,\tfrac{1}{2}\|\rho-\sigma\|$, the fidelity deficit $\,1-F(\rho,\sigma)$, the Bures distance $\,\beta(\rho,\sigma)$,
etc.).

The von Neumann entropy $S$ (the entropy, in what follows) is the main
characteristic of a quantum state $\rho$ \cite{H-SCI,W,Wilde}. It is defined by the formula
\begin{equation}\label{ent-def}
S(\rho)=\operatorname{Tr}\eta(\rho),
\end{equation}
where  $\eta(x)=-x\ln x$ if $x>0$ and $\eta(0)=0$.

The importance of the entropy explains a
number of works devoted to quantitative continuity analysis of this
characteristic in both finite and infinite-dimensional quantum systems
(see the brief overview in  \cite[the Introduction]{BDJSH}).

To solve one specific problem (described in \cite{H&Sh-5}), we needed to get a lower semicontinuity bound for the entropy
at a given state $\rho$ of a  quantum system described by a finite-dimensional Hilbert space
$\mathcal{H}$ of the form
\begin{equation}\label{b2}
S(\rho)-S(\sigma)\leq C_\rho \beta^2(\rho,\sigma),\quad \forall\sigma\in \St(\mathcal H),
\end{equation}
where  $\beta(\rho,\sigma)$ is the Bures distance between $\rho$ and $\sigma$ and $\St(\mathcal H)$ denotes
the set of all states on $\HH$. It turned out that this cannot be done for each state $\rho$, since  one can show that
\[
  \sup_{\sigma \in \mathfrak{S}(\mathcal{H})}
  \frac{S(\rho) - S(\sigma)}{\beta^{2}(\rho,\sigma)} = +\infty
\]
for some states $\rho$ in $\St(\mathcal H)$. Nevertheless, the log-Sobolev inequality allows us to show that the above supremum is finite
for all finite-rank states $\rho$ of finite and infinite-dimensional quantum systems
which has uniform positive spectrum. In the first subsection of the main part of this article we show that inequality (\ref{b2})
holds for a state $\rho$ with the  spectrum
\begin{equation}\label{ud}
\Bigl(\underbrace{\frac{1}{d},\frac{1}{d},...,\frac{1}{d}}_{d\textrm{ items}},0,0...\Bigr)
\end{equation}
with $C_{\rho}=K_{d}$ -- the optimal constant in the log-Sobolev inequality for the complete graph with $d$ vertices. Moreover, the
stronger version of this inequality with $\beta^{2}(\rho,\sigma)$ replaced by $(1-F(\rho,\sigma))$ is valid, and in this form the
inequality becomes sharp (optimal).\smallskip

The Entanglement of Formation (EoF) is one of the basic entanglement measures in bipartite quantum systems \cite{P&V,4H}. In general infinite-dimensional settings
the EoF of a state $\rho$ of a composite system $AB$ is defined as the closed (lower semicontinuous) convex  roof extension of the function
$\rho\mapsto S(\rho_A)=S(\rho_B)$ on the set $\ext\St(\HH_{AB})$ of pure states of the system $AB$ to the set $\St(\HH_{AB})$ of all states of this system, which
can be expressed as
\begin{equation}\label{EF-c}
E_F(\rho)=\!\inf_{\int\varrho\mu(d\varrho)=\rho}\int_{\ext\St(\HH_{AB})}S(\varrho_A)\mu(d\varrho),
\end{equation}
where the  infimum is over all Borel probability measures\footnote{It is reasonable to consider such measures as generalized ensembles of quantum states \cite{H-Sh-2,H-SCI}.} on the set of pure states in $\St(\HH_{AB})$ with the barycenter $\rho$ \cite{EM, Lami-new}. If $A$ and $B$ are finite-dimensional quantum systems then the expression (\ref{EF-c}) is reduced to the well-known formula
\begin{equation}\label{EF-d}
E_F(\rho)=\inf_{\sum_k\!p_k\rho_k=\rho}\sum_kp_kS([\rho_k]_A),
\end{equation}
where the  infimum is over all finite  ensembles $\{p_k, \rho_k\}$ of pure states in $\St(\HH_{AB})$ with the average state $\rho$ \cite{Bennett}.
\smallskip

The second part of this  note is devoted to the problem of finding, for a given state
$\,\rho\in \St(\HH_{AB}),$  estimates of the form
\begin{equation}\label{one}
    E_F(\rho)-E_F(\sigma)
    \leq
    C_\rho\,\frac{\|\rho-\sigma\|_1}{2},
    \qquad \forall\,\sigma \in \St(\HH_{AB}),
\end{equation}
where $C_\rho$ is a constant (as small as possible) depending on $\rho$ and not depending on
$\sigma$.

It is shown in \cite{H&Sh-5} that the validity of (\ref{one}) with some $C_\rho>0$ is
equivalent to the existence of a supporting affine functional  $\,\ell_\rho(\sigma)=\Tr\Lambda_\rho \sigma\,$  for the EoF
at the state $\rho$ with Hermitian operator
$\Lambda_\rho$ whose spectral diameter does not exceed $C_\rho$.\footnote{This means that $\,\ell_\rho(\rho)=E_F(\rho)\,$ and $\,\ell_\rho(\sigma)\leq E_F(\sigma)\,$ for any state $\sigma$ in $\St(\HH_{AB})$.}

It is also shown in \cite{H&Sh-5} that in any bipartite quantum systems (in particular, in
 $2$-qubit system) there exist states $\rho$ such that
\begin{equation}
    \sup_{\sigma\in \St(\HH_{AB})}
    \frac{E_F(\rho)-E_F(\sigma)}
         {\|\rho-\sigma\|_1}
    =+\infty,
\end{equation}
which means that (\ref{one}) cannot be valid with any finite $C_\rho$.

It is somewhat surprising that the above problem is nontrivial even in the
case when the state $\rho$ is pure. We and GPT-5.6 have found neither a general proof
nor counterexamples which would confirm or refute the hypothesis of the existence of a supporting affine functional  for the EoF
at any pure bipartite state.\footnote{Using Wootter's formula and the help of GPT-5.6, we tried to build an example showing the lack of supporting affine functional for the EoF at some pure state of $2$-qubit system (similar to the example presented in \cite{H&Sh-5}). However, GPT-5.6 was not possible to find such an example and  gave a non-strict proof that such an example does not exist at all.}

Fortunately, there is an important class of bipartite pure states for which
the above problem has a simple sharp solution. This class consists of pure states whose set of  Schmidt numbers forms the distribution (\ref{ud}) for some $d\geq 2$, (i.e. such
states whose marginals has the spectrum (\ref{ud})).

It is remarkable that this sharp  solution turned out to be closely related to the sharp log-Sobolev inequality for a complete graph. Moreover,
it turns out that the optimal constant $C_{\rho}$ in (\ref{one}) in the case when  $\rho$ is a pure state with the Schmidt numbers in (\ref{ud}) coincides the optimal constant $K_{d}$ in the log-Sobolev inequality for complete graph with $d$ vertices.

\section{Preliminaries}

Throughout the article, we assume that $\mathcal{H}_X$ is a separable Hilbert space describing a quantum system $X$,
$\mathfrak{B}_{\rm sa}(\mathcal{H}_X)$ is the real Banach space of all Hermitian bounded operators on $\mathcal{H}_X$ with the operator norm $\|\cdot\|$ and $\mathfrak{T}(\mathcal{H}_X)$ is the (complex)
Banach space of all trace-class
operators on $\mathcal{H}_X$  with the trace norm $\|\!\cdot\!\|_1$. Let
$\mathfrak{S}(\mathcal{H}_X)$ be  the set of quantum states (positive operators
in $\mathfrak{T}(\mathcal{H}_X)$ with unit trace) \cite{H-SCI,Wilde}.\smallskip

Write $I_{X}$ for the unit operator on a Hilbert space
$\mathcal{H}_X$.\smallskip

The \emph{von Neumann entropy} of a quantum state
$\rho \in \mathfrak{S}(\HH)$  (defined by  formula (\ref{ent-def}))  is a concave lower semicontinuous function on the set~$\mathfrak{S}(\HH)$ taking values in~$[0,+\infty]$ \cite{H-SCI,L-2,W}.

We will use the  homogeneous extension of the von Neumann entropy to the positive cone $\T_+(\HH)$ defined as
\begin{equation}\label{S-ext}
S(\rho)\doteq(\Tr\rho)S(\rho/\Tr\rho)=\Tr\eta(\rho)-\eta(\Tr\rho)
\end{equation}
for any nonzero operator $\rho$ in $\T_+(\HH)$ and equal to $0$ at the zero operator \cite{L-2}.\smallskip

The concavity of the von Neumann entropy on $\St(\HH)$ implies that
\begin{equation}\label{S-in}
S(\rho)+S(\sigma)\leq S(\rho+\sigma)\quad \forall \rho,\sigma\in\T_+(\HH).
\end{equation}

We will use the sharp log-Sobolev inequality for the uniform measure on a complete graph obtained
in \cite{LS-1}, Theorem A.1. This inequality states that for a real random variable $f$ taking values $f_1,\dots,f_d$
with equal probabilities $1/d$,
\begin{equation}\label{l-S}
\mathbb{E}\!\left[f^2\ln f^2\right]
-
\mathbb{E}[f^2]\ln\mathbb{E}[f^2]
\leq
K_d\,\mathcal{E}(f,f),
\end{equation}
where \begin{equation}\label{Cdd+}
K_d=
\begin{cases}
2, & d=2,\\[2mm]
\dfrac{d}{d-2}\ln(d-1), & d>2,
\end{cases}
\end{equation}
and
\[
\mathcal{E}(f,f)
=\frac{1}{2d^2}
\sum_{i,j=1}^{d}(f_i-f_j)^2=\frac{1}{2d^2}\left(2d^2-2\left(\sum_{k=1}^{d}f_k\right)^2\right),
\]
as follows from the general definition of the Dirichlet form (2.2) in \cite{LS-1} applied to the case of Theorem A.1.

Thus  for any non-negative random variable $f$ such that $\mathbb{E}[f^2]=1$,
the log-Sobolev inequality (\ref{l-S}) takes the form
\begin{equation}\label{l-S+}
\frac1d\sum_i f_i^2\ln f_i^2
\;\le\;
K_d\Bigl[1 - (\tfrac1d\sum_i f_i)^2\Bigr].
\end{equation}
By setting $\lambda_i=f^2_i/d$, $i=\overline{1,d}$, we obtain from (\ref{l-S+}) the inequality
\begin{equation}\label{ll}
\ln d + \sum_{i=1}^d\lambda_i\ln\lambda_i
\;\le\;
K_d\left[1 - \frac1d\Bigl(\sum_{i=1}^d\sqrt{\lambda_i}\Bigr)^2\right]
\end{equation}
valid for any probability distribution $\{\lambda_i\}_{i=1}^d$  with $d$ outcomes.
We would like to notice here that this inequality arose independently in completely different context,
as a special case of optimality conditions for accessible information of a quantum ensemble
(see \cite{HU25}, Eq. (26)).

\section{Main results}

\subsection{Lipschitz  semicontinuity bound for the von Neumann entropy w.r.t. the fidelity deficit $(1-F(\rho,\sigma))$}

In our recent article \cite{H&Sh-5} we discussed the problem of existence of a supporting
affine functional for the EoF at a state of  finite and
infinite-dimensional composite quantum systems. One of the questions which are
not resolved in \cite{H&Sh-5} concerns the existence of supporting
affine functional at any pure state of a
finite-dimensional composite quantum system. It is not hard to show that this question has a positive solution for a
pure state $\rho$ of a bipartite system $AB$ provided that there is a finite number $C_{\rho} > 0$ such that
\[
  S(\rho) - S(\sigma) \le C_{\rho}\,\beta^{2}(\rho,\sigma),
  \qquad \forall\, \sigma \in \mathfrak{S}(\mathcal{H}_A),
\]
where $\,\beta(\rho,\sigma)\,$ is the Bures distance between  $\rho$ and $\sigma$  \cite{H-SCI, Wilde}.  Unfortunately, we cannot claim  the existence of such number $C_{\rho}$ for
all pure states $\rho$ in $\mathfrak{S}(\mathcal{H}_{AB})$, since one
can show that
\[
  \sup_{\sigma \in \mathfrak{S}(\mathcal{H})}
  \frac{S(\rho) - S(\sigma)}{\beta^{2}(\rho,\sigma)} = +\infty
\]
for some states $\rho$ of a quantum system described by a finite-dimensional Hilbert space
$\mathcal{H}$.\footnote{It is known that the von Neumann entropy is Lipschitz-continuous w.r.t. the Bures
distance \cite{LB}, but it is not Lipschitz-continuous (and is not Lipschitz-lower-semicontinuous) w.r.t. the squared Bures distance.}

Nevertheless, there is a class of states for which the above supremum is finite.

\begin{proposition}\label{main-2} Let $\HH$ be a Hilbert space of any dimension and $\rho$ be a state in $\St(\HH)$
of rank $d$ with the spectral representation
\begin{equation*}
\rho=\frac{1}{d}\sum_{k=1}^{d}|\phi_k\rangle\langle\phi_k|,\quad d\geq 2
\end{equation*}
where $\{\phi_k\}_{k=1}^{d}$ is an orthonormal system in $\HH$. Then
\begin{equation}\label{main+2}
 S(\rho)-S(\sigma)\leq K_d (1-F(\rho,\sigma))\leq K_d \beta^2(\rho,\sigma),
\end{equation}
for any state $\sigma\in \St(\mathcal H)$, where
\begin{equation}\label{Cdd}
K_d=
\begin{cases}
2, & d=2,\\[2mm]
\dfrac{d}{d-2}\ln(d-1), & d>2,
\end{cases}
\end{equation}
is the optimal constant in the log-Sobolev inequality (\ref{l-S}), $\,F(\rho,\sigma)=
\left\|\sqrt{\rho}\sqrt{\sigma}\right\|^2\,$ and  $\beta(\rho,\sigma)$
are the fidelity  and the Bures distance between $\rho$ and $\sigma$, respectively.\smallskip

The first inequality in (\ref{main+2}) is sharp (optimal):
\[
\sup_{\sigma\ne\rho}\shs
\frac{S(\rho)-S(\sigma)}
{1-F(\rho,\sigma)}
=K_d.
\]
If $\,d\geq3\,$ then this supremum is attained at the state
\begin{equation*}
\sigma=\sum_{k=1}^{d}\lambda_k|\phi_k\rangle\langle\phi_k|,
\end{equation*}
where $\,\lambda_1=\tfrac{d-1}d\,$ and $\,\lambda_2=..=\lambda_d=\tfrac1{d(d-1)}.$
\end{proposition}

\emph{Proof.}  W.l.o.g. we may assume that $\dim\HH=+\infty$. We have to prove the first inequality (\ref{main+2}), because
the second one follows from the definitions of the fidelity  and the Bures distance.

Assume first that $\sigma$ is a state of rank $d$ with the spectrum $\,(\lambda_1,...,\lambda_d,0,0,...)\,$ such that
$\supp\sigma\subseteq\supp\rho$. Then inequality (\ref{ll}) implies directly  that
\begin{equation*}
 S(\rho)-S(\sigma)\leq K_d (1-\|\sqrt{\rho}\sqrt{\sigma}\|_1^2)=K_d(1-F(\rho,\sigma)).
\end{equation*}

Assume now  that $\sigma$ is an arbitrary state in $\St(\HH)$
with the spectral representation
\begin{equation*}
\sigma=\sum_{k=1}^{+\infty}\lambda_k|\psi_k\rangle\langle\psi_k|,\quad \lambda_k\geq\lambda_{k+1}\geq 0\;\; \forall k,
\end{equation*}
where $\{\psi_k\}_{k=1}^{+\infty}$ is an orthonormal system in $\HH$. Consider the state
\begin{equation*}
\hat{\sigma}=\sum_{k=1}^{+\infty}\lambda_k|\phi_k\rangle\langle\phi_k|,
\end{equation*}
where  $\{\phi_k\}_{k=1}^{+\infty}$ is an orthonormal system obtained by a particular  extension of the system $\{\phi_k\}_{k=1}^{d}$. It is clear that
\begin{equation}\label{1s}
S(\sigma)=S(\hat{\sigma})\quad \textrm{and}\quad F(\rho,\hat{\sigma})=\frac{1}{d}\shs[\Tr\sqrt{\sigma}]^2\geq\frac{1}{d}\shs[\Tr P_{\rho}\sqrt{\sigma}]^2=F(\rho,\sigma),
\end{equation}
where $P_{\rho}=d\rho$ is a projector onto the subspace $\supp\rho$. \smallskip

Consider the quantum channel
$$
\Phi(\varrho)=\sum_{k=0}^{+\infty}W_k\varrho W_k^*,
$$
where $W_k=\sum_{i=1}^{d}|\phi_i\rangle\langle\phi_{i+k}|$ is a partial isometry such that
$W_kW^*_k$ and $W^*_kW_k$ are the projectors on the linear spans of $\{\phi_1,...,\phi_d\}$ and $\{\phi_{1+k},...,\phi_{d+k}\}$
correspondingly.

Let $\tilde{\sigma}=\Phi(\hat{\sigma})$. Since $\rho=\Phi(\rho)$ by the construction, we have
\begin{equation}\label{2s}
F(\rho,\tilde{\sigma})=F(\Phi(\rho),\Phi(\hat{\sigma}))\geq F(\rho,\hat{\sigma})
\end{equation}
due to monotonicity of the fidelity under action of a channel. Thus, as $\supp\tilde{\sigma}\subseteq\supp\rho$,
to prove the
first inequality in (\ref{main+2}) by combining (\ref{1s}), (\ref{2s}) and the first part of this proof it suffices to show that $S(\tilde{\sigma})\leq S(\sigma)$. This can be done by using the Schur concavity of the von Neumann entropy and by noting that
$$
\tilde{\sigma}=\sum_{k=0}^{+\infty}\lambda_{1+k}|\phi_1\rangle\langle\phi_1|+...+\sum_{k=0}^{+\infty}\lambda_{d+k}|\phi_d\rangle\langle\phi_d|,
$$
and hence the state $\sigma$ is majorized by the state $\tilde{\sigma}$ in the sense of \cite[Section 13.5]{S-T-1}.
\smallskip

The optimality claim for $\,d\geq 3\,$ is proved by direct calculation:
$$
\ln d-S(\sigma)=\frac{d-2}{d}\ln(d-1)= K_d\left(1-\frac{1}{d}\shs[\Tr\sqrt{\sigma}]^2\right).
$$
For $d=2$ we cannot find "optimal" state $\sigma$. So, consider the family of states
\[
\sigma_\varepsilon=
\begin{pmatrix}
\frac12+\varepsilon&0\\
0&\frac12-\varepsilon
\end{pmatrix},
\qquad
\varepsilon\in(0,\textstyle\frac{1}{2}).
\]
Then
\[
S(\sigma_\varepsilon)
=
-\left(\frac12+\varepsilon\right)
 \ln\left(\frac12+\varepsilon\right)
-\left(\frac12-\varepsilon\right)
 \ln\left(\frac12-\varepsilon\right).
\]
A Taylor expansion at $\varepsilon=0$ gives
$\, \ln2-S(\sigma_\varepsilon)
=
2\varepsilon^2+O(\varepsilon^4)$. On the other hand,
\[
F(\rho,\sigma_\varepsilon)
=
\left(
\sqrt{\frac12\left(\frac12+\varepsilon\right)}
+
\sqrt{\frac12\left(\frac12-\varepsilon\right)}
\right)^2
=
\frac12+\sqrt{\frac14-\varepsilon^2}.
\]
Hence $\,1-F(\rho,\sigma_\varepsilon)
=\varepsilon^2+O(\varepsilon^4)$. Therefore
\[
\lim_{\varepsilon\to0^+}\frac{S(\rho)-S(\sigma_\varepsilon)}
{1-F(\rho,\sigma_\varepsilon)}
=2.
\]

$\Box$ \smallskip

\subsection{Lipschitz  semicontinuity bound for the EoF}

The main result of this note is the following

\begin{proposition}\label{main}
Let $A$ and $B$ be quantum systems of any dimensions. Let $\rho$
be a pure state of the system $AB$ with the Schmidt representation
\begin{equation}\label{smr}
\rho
=
\frac{1}{d}
\sum_{i,j=1}^{d}
|\varphi_i\rangle\langle\varphi_j|
\otimes
|\psi_i\rangle\langle\psi_j|,\quad d\geq 2,
\end{equation}
where $\{\varphi_1,\ldots,\varphi_d\}$ and
$\{\psi_1,\ldots,\psi_d\}$ are orthonormal systems of vectors in
$\mathcal H_A$ and $\mathcal H_B$, respectively. Then
\begin{equation}\label{main+}
E_F(\rho)-E_F(\sigma)
\leq
K_d\bigl(1-\operatorname{Tr}\rho\sigma\bigr)
\leq
K_d\frac{\|\rho-\sigma\|_1}{2},
\end{equation}
for any state $\sigma\in \St(\mathcal H_{AB})$, where
\begin{equation}\label{Cdd++}
K_d=
\begin{cases}
2, & d=2,\\[2mm]
\dfrac{d}{d-2}\ln(d-1), & d>2,
\end{cases}
\end{equation}
is the optimal constant in the log-Sobolev inequality (\ref{l-S}).\smallskip

The first inequality in (\ref{main+}) is sharp (optimal):
\[
\sup_{\sigma\ne\rho}\shs
\frac{E_F(\rho)-E_F(\sigma)}
{1-\Tr\rho\sigma}
=K_d.
\]
If $\,d\geq3\,$ then this supremum is attained at the pure state
\begin{equation}\label{smr+}
\sigma
=
\sum_{i,j=1}^{d}\sqrt{\lambda_i\lambda_j}\;
|\varphi_i\rangle\langle\varphi_j|
\otimes
|\psi_i\rangle\langle\psi_j|,
\end{equation}
where $\,\lambda_1=\tfrac{d-1}d\,$ and $\,\lambda_2=..=\lambda_d=\tfrac1{d(d-1)}$.
\end{proposition}

\begin{remark}\label{main-r}
By Lemma \ref{sl} below $a\doteq K_d-\ln d>0$ for any natural $d\geq 2$. The first inequality in (\ref{main+}) shows that the functional
$\,\ell_\rho(\sigma)=\Tr\Lambda_{\rho}\sigma\,$
on $\St(\HH_{AB})$, where
$$\Lambda_{\rho}\doteq K_d\rho-aI_{AB}=(\ln d)\rho\ominus a(I_{AB}-\rho),$$
is a supporting functional for the EoF at the state $\rho$, which means that
\begin{equation}\label{gsf}
\ell_\rho(\rho)= E_F(\rho)=\ln d\quad \textrm{and} \quad \ell_\rho(\sigma)\leq E_F(\sigma)\quad \forall\sigma\in\St(\HH_{AB}).
\end{equation}
The spectrum of the  above operator $\Lambda_{\rho}$ is discrete, the minimal and the maximal eigenvalues of $\Lambda_{\rho}$ are equal to $\,-a\,$  and $\,\ln d\,$, respectively. Hence, the spectral diameter of operator $\Lambda_{\rho}$  is equal to $K_d$.

This should be compared with the fact that Theorem 1 in \cite{H&Sh-5} and the second inequality in (\ref{main+}) (easily derived from the first one)  imply the existence of a supporting functional $\,\ell_\rho(\sigma)=\Tr\Lambda_{\rho}\sigma\,$  for the EoF at the state $\rho$ with the Hermitian operator $\Lambda_{\rho}$ having the  spectral diameter $K_d$. This leads us to the open question formulated at the end of the article.

Note, finally, that the relations in (\ref{gsf}) for the functional $\ell_\rho$ directly implies the first inequality in (\ref{main+}).
\end{remark}

By noting that $\,\Tr \rho\shs\sigma=\left\|\sqrt{\rho}\sqrt{\sigma}\right\|^2=F(\rho,\sigma)\,$ for any pure state $\rho$
and arbitrary state $\sigma$ we may reformulate Proposition \ref{main} as follows.

\begin{corollary}\label{main-c}
Let $A$ and $B$ be quantum systems of any dimensions. Let $\rho$ be a pure state of the system $AB$ with the Schmidt representation
(\ref{smr}). Then
\begin{equation}\label{main-c+}
E_F(\rho)-E_F(\sigma)
\leq K_d\bigl(1-F(\rho,\sigma)\bigr)
\leq K_d\,\beta^2(\rho,\sigma),
\end{equation}
where  $K_d$ is defined in (\ref{Cdd}), $\,F(\rho,\sigma)=
\left\|\sqrt{\rho}\sqrt{\sigma}\right\|^2\,$ and  $\,\beta(\rho,\sigma)$
are the fidelity  and the Bures distance between $\rho$ and $\sigma$, respectively.

The first inequality in (\ref{main-c+}) is sharp (optimal):
\[
\sup_{\sigma\ne\rho}\shs
\frac{E_F(\rho)-E_F(\sigma)}
{1-F(\rho,\sigma)}
=K_d.
\]
If $\,d\geq3\,$ then this supremum is attained at the pure state $\sigma$ defined in (\ref{smr+}).
\end{corollary}

\emph{Proof of Proposition \ref{main}.}  Assume first that $\sigma$ is a pure state with
the Schmidt
numbers $\,\left(\lambda_1,\ldots,\lambda_d,0,0,...\right)\,$
such that $\,\supp \sigma_A\subseteq \HH_A^\rho\doteq\supp \rho_A$ and $\,\supp \sigma_B\subseteq \HH_B^\rho\doteq\supp \rho_B$.

Since
$$
\frac1d\Bigl(\sum_{i=1}^d\sqrt{\lambda_i}\Bigr)^2=\left[\Tr\sqrt{\sqrt{\sigma_A}\rho_A\sqrt{\sigma_A}}\right]^2=F(\rho_A,\sigma_A)\geq F(\rho,\sigma),
$$
inequality (\ref{ll}) implies
\begin{equation}\label{lll}
\ln d - S(\sigma_A)
\;\le\;
K_d\;[1 - F(\rho,\sigma)]\;=\;
K_d\;[1 - \Tr\rho\shs\sigma].
\end{equation}

Let $\sigma$ be an arbitrary mixed state supported by the subspace $\,\HH_A^\rho\otimes\HH_B^\rho$.
Let $\nu$ be an optimal measure for the state $\sigma$, i.e. a  probability measure  $\nu$ on $\,\ext\St(\HH_A^\rho\otimes\HH_B^\rho)\,$ such that
\[
\sigma=\int_{\ext\St(\HH_A^\rho\otimes\HH_B^\rho)} \psi\,\nu(d\psi)
\qquad \textrm{and}\qquad E_F(\sigma)
=
\int_{\ext\St(\HH_A^\rho\otimes\HH_B^\rho)} S(\psi_A)\,\nu(d\psi).
\]
Then it follows from (\ref{lll}) that
\[
\begin{aligned}
\ln d-E_F(\sigma)
&=
\int_{\ext\St(\HH_A^\rho\otimes\HH_B^\rho)}
\bigl[\ln d - S(\psi_A)\bigr]\,\nu(d\psi)
\\
&\leq
K_d
\int_{\ext\St(\HH_A^\rho\otimes\HH_B^\rho)}
\bigl[1-\operatorname{Tr}(\rho\psi)\bigr]\,\nu(d\psi)=K_d\bigl[1-\operatorname{Tr}(\rho\shs\sigma)\bigr].
\end{aligned}
\]

Let $\,\Lambda^0_\rho\doteq K_d\rho-(K_d-\ln d)I_\rho\,$ be an operator on $\,\HH_A^\rho\otimes\HH_B^\rho$, where $I_\rho$
denotes the unit operator on $\,\HH_A^\rho\otimes\HH_B^\rho$, then the last inequality shows that the functional
\[
\ell^0_{\rho}(\varrho)=\Tr\Lambda^0_\rho\varrho,\quad  \varrho\in\St(\HH_A^\rho\otimes\HH_B^\rho),
\]
is a supporting functional for the restriction of the EoF to the set  $\,\St(\HH_A^\rho\otimes\HH_B^\rho)\,$ (-- the face of $\,\St(\HH_A\otimes\HH_B)!$) at the state $\rho$, i.e.
\begin{equation}\label{lsf}
\ell^0_{\rho}(\rho)=E_F(\rho)=\ln d\quad \textrm{and} \quad \ell^0_{\rho}(\sigma)\leq E_F(\sigma)\quad \forall\sigma\in\St(\HH_A^\rho\otimes\HH_B^\rho).
\end{equation}
Using Lemma \ref{rl} and Remark \ref{rlr} below it is easy to show that the functional $\,\ell_\rho(\sigma)=\Tr\Lambda_{\rho}\sigma,$ where
\begin{equation}\label{old}
\Lambda_{\rho}\doteq K_d\rho-aI_{AB},\quad  a=K_d-\ln d>0,
\end{equation}
is a (global) supporting functional for the EoF  at the state $\rho$, which means the validity of the relations in (\ref{gsf}). The second of these relations implies the first inequality in (\ref{main+}).

The second inequality in (\ref{main+}) can be derived from the first one by using the well known continuity bound for the function
$\,\vartheta\mapsto\Tr\rho\vartheta\,$ (see, f.i., Example 2 in \cite{QC}), since the spectral diameter of operator $\rho$ is equal to $1$.\smallskip

The optimality claim for $\,d\geq 3\,$ is proved by direct calculation:
$$
\ln d-S(\sigma_A)=\frac{d-2}{d}\ln(d-1)= K_d\left(1-\Tr\rho\sigma\right).
$$
For $d=2$ we cannot find "optimal" state $\sigma$. So, consider the pure states with Schmidt coefficients
\[
\sqrt{\frac12+\varepsilon},
\quad
\sqrt{\frac12-\varepsilon},\qquad  \varepsilon\in(0,\textstyle\frac{1}{2}).
\]
For these states, $\,E_F(\sigma_\varepsilon)
=h\!\left(\frac12+\varepsilon\right)$, where $h(x)=-x\ln x-(1-x)\ln(1-x)$ is the binary entropy. Therefore $\,E_F(\rho)-E_F(\sigma_\varepsilon)
=2\varepsilon^2+O(\varepsilon^4)\,$.

Since both states are pure,
$\,\operatorname{Tr}(\rho\sigma_\varepsilon)=F(\rho,\sigma_\varepsilon)
=1-\varepsilon^2+O(\varepsilon^4)$. Consequently,
\[
\lim_{\varepsilon\to0^+}\frac{
E_F(\rho)-E_F(\sigma_\varepsilon)
}{
1-\operatorname{Tr}(\rho\sigma_\varepsilon)}=2.
\]
$\Box$\smallskip

The proof of Proposition \ref{main} is based on the following lemma, in which we write $D(\Lambda)$ for the diameter of the spectrum $\mathrm{Sp}(\Lambda)$ of an operator $\Lambda\in\B_{\rm sa}(\HH)$:
\begin{equation}\label{D-def}
D(\Lambda)=\max \mathrm{Sp}(\Lambda)-\min \mathrm{Sp}(\Lambda).
\end{equation}

\begin{lemma}\label{rl} Let $A$ and $B$ be quantum systems of any dimensions. Let $\rho$ be a state of the system $AB$, $\,\HH^\rho_X\doteq\supp \rho_X\,$ and $\,P_\rho^X\in\Bsa(\HH_X)$ be the
projector onto $\HH_X^\rho$, $X=A,B$.

If there  exists an operator
$\,\Lambda_\rho\in \B_{\rm sa}(\HH_A^\rho\otimes\HH_B^\rho )\,$ such that
\begin{equation}\label{con}
E_F(\rho)=\Tr\Lambda_\rho\rho
\quad\; \textit{and}\quad\;
\langle\psi|\Lambda_\rho|\psi\rangle\leq S(\psi_A)
\quad \forall\,\psi\in \HH^A_\rho\otimes\HH^B_\rho,\; \|\psi\|=1,
\end{equation}
and $\,c\geq 0\,$ is a given number, then
\begin{equation}\label{ent}
E_F(\rho)=\Tr\Lambda^{\rm ext}_{\rho,c}\rho
\quad\; \textit{and}\quad\;
\langle\psi|\Lambda^{\rm ext}_{\rho,c}|\psi\rangle\leq S(\psi_A)
\quad \forall\,\psi\in \HH_{AB},\, \|\psi\|=1,
\end{equation}
where
$$
\Lambda^{\rm ext}_{\rho,c}=P_\rho^A\otimes P_\rho^B\cdot\Lambda_\rho\cdot P_\rho^A\otimes P_\rho^B-c\left(I_{AB}-P_\rho^A\otimes P_\rho^B\right)
$$
is a Hermitian operator on $ \HH_{AB}$ such that\footnote{$\rm Sp(\Lambda_\rho)$ and $D(\Lambda_\rho) $ are the spectrum of $\Lambda_\rho$ and its diameter (defined in (\ref{D-def})).}
$$
D(\Lambda^{\rm ext}_{\rho,c})=D(\Lambda_\rho)+\max\{0, c+\inf\rm Sp(\Lambda_\rho)\}.
$$
\end{lemma}

\begin{remark}\label{rlr}
It is easy to see that the second property in (\ref{con})  (resp. in (\ref{ent})) holds if and only if  $\,\Tr\Lambda_{\rho}\sigma\leq E_F(\sigma)\,$ for all $\sigma$ in $\St(\HH^A_\rho\otimes\HH^B_\rho)$  (resp. $\,\Tr\Lambda^{\rm ext}_{\rho,c}\sigma\leq E_F(\sigma)\,$ for all $\sigma$ in $\St(\HH_{AB})$).
\end{remark}

\emph{Proof.} The first equality in (\ref{ent})  is obvious. To
show that $\,\langle\psi|\Lambda_{\rho,c}^{\ext}|\psi\rangle\leq S(\psi_A)\,$ for all $\,\psi\in \HH^1_{AB}$
denote $P_\rho^A\otimes P_\rho^B$ by $Q_\rho$ and note that
\begin{align*}
[Q_\rho\psi]_A
&=
\Tr_B
\bigl(
P_\rho^A\otimes P_\rho^B
|\psi\rangle\langle\psi|
P_\rho^A\otimes P_\rho^B
\bigr)
\\
&=
P_\rho^A
\Bigl[
\Tr_B
\bigl(
I_A\otimes P_\rho^B
|\psi\rangle\langle\psi|
\bigr)
\Bigr]
P_\rho^A
\\
&\leq
P_\rho^A
\bigl[
\Tr_B|\psi\rangle\langle\psi|
\bigr]
P_\rho^A,
\end{align*}
where $I_A$ is the unit operator on $\HH_A$. Hence inequality (\ref{S-in}) and Lemma 3 in \cite{L-2} show that
\[
S\bigl([Q_\rho\psi]_A\bigr)
\leq S\bigl(P_\rho^A
\bigl[
\Tr_B|\psi\rangle\langle\psi|
\bigr]
P_\rho^A\bigr)
\leq S(\psi_A)
\qquad \forall\,\psi\in \HH^1_{AB},
\]
where $S$ in the l.h.s. is the extension
of the von Neumann entropy defined in (\ref{S-ext}). Thus, we have
\[
\langle\psi|\Lambda_{\rho,c}^{\ext}|\psi\rangle
=\langle Q_\rho\psi|\Lambda_\rho|Q_\rho\psi\rangle
-c\langle \psi|I_{AB}-Q_\rho|\psi\rangle \leq S\bigl([Q_\rho\psi]_A\bigr)
\leq S(\psi_A),\qquad \forall\,\psi\in \HH^1_{AB},
\]
where the first inequality follows from the second property in (\ref{con}). $\Box$
\medskip

The following lemma is proved using simple arguments and the well-known estimates for the logarithmic function.

\begin{lemma}\label{sl}  For any natural $\,d\geq 2\,$ the optimal constant $K_d$ in the log-Sobolev inequality (\ref{l-S})  defined in (\ref{Cdd}) is greater than $\,\ln d$.
\end{lemma}

\subsection{Improved version of Proposition \ref{main}}

The estimates in Proposition \ref{main} can be  improved using the information about $\Tr P_\rho\sigma$,
where $P_\rho$ is the projector onto the subspace $\,\supp\rho_A\otimes\supp\rho_B$. \smallskip

\begin{proposition}\label{main++}
Let $A$ and $B$ be quantum systems of any dimensions. Let $\rho$ be a pure state of the system $AB$ with the Schmidt representation
(\ref{smr}). Let $P_\rho$ be the projector onto the subspace $\,\supp\rho_A\otimes\supp\rho_B$ and $\,P^\bot_\rho\doteq I_{AB}-P_\rho$.  Then
\begin{equation}\label{main+++}
E_F(\rho)-E_F(\sigma)
\leq
K_d\bigl(1-\operatorname{Tr}\rho\sigma\bigr)-a\Tr P^\bot_\rho\sigma
\leq
K_d\frac{\|\rho-\sigma\|_1}{2}-a\Tr P^\bot_\rho\sigma
\end{equation}
for any state $\sigma$ in $\St(\HH_{AB})$, where  $K_d$ is defined in (\ref{Cdd++}) and $\,a\doteq K_d-\ln d>0$.

The first inequality in (\ref{main+++}) is sharp (optimal):
\[
\sup_{\sigma\ne\rho}\shs
\frac{E_F(\rho)-E_F(\sigma)}
{K_d\bigl(1-\operatorname{Tr}\rho\sigma\bigr)-a\Tr P^\bot_\rho\sigma}
=1.
\]
If $\,d\geq3\,$ then this supremum is attained at the pure state $\sigma$ defined in (\ref{smr+}).
\end{proposition}

\emph{Proof.} To prove the first inequality in (\ref{main+++}) it suffices
to change the arguments from the proof of Proposition \ref{main} after (\ref{lsf}) as follows.

Lemma \ref{rl} and Remark \ref{rlr} (applied to the operator $\Lambda^0_{\rho}$) show that the functional $\,\tilde{\ell}_\rho(\sigma)=\Tr\widetilde{\Lambda}_{\rho}\sigma,$ where
\begin{equation}\label{new}
\widetilde{\Lambda}_{\rho}\doteq K_d\rho-a P_\rho,
\end{equation}
is a (global) supporting functional for the EoF  at the state $\rho$. This does not contradict the fact that the functional
$\,\ell_\rho(\sigma)=\Tr\Lambda_{\rho}\sigma\,$ (where $\Lambda_{\rho}$ is the operator defined in (\ref{old})) is a (global) supporting functional for the EoF  at the state $\rho$ (because a supporting functional is not unique, in general). Then the relations in (\ref{gsf}) hold with $\ell_\rho$ replaced by $\tilde{\ell}_\rho$. The second of them shows that
$$
K_d\Tr\rho\sigma-a\Tr P_\rho\sigma\leq E_F(\sigma).
$$
Since $\,E_F(\rho)=\ln d,$ this inequality implies the first inequality in (\ref{main+++}).\smallskip

Another way to prove the first inequality in (\ref{main+++}) is to derive it from the first inequality in (\ref{main+}) by using the
selective LOCC-monotonicity of the EoF (via the measurement induced by the projector $P_\rho$).\smallskip

The second inequality in (\ref{main+++}) can be derived from the first one by using the well known continuity bound for the function
$\,\vartheta\mapsto\Tr\rho\vartheta\,$ (see, f.i., Example 2 in \cite{QC}), since the spectral diameter of operator $\rho$ is equal to $1$. $\Box$
\medskip

The same improvement can be done in the estimates in Corollary \ref{main-c}: the terms
$\,K_d\bigl(1-F(\rho,\sigma)\bigr)\,$ and $\,K_d\,\beta^2(\rho,\sigma)\,$ in (\ref{main-c+}) can be replaced with
$$
K_d\bigl(1-F(\rho,\sigma)\bigr)-a\Tr P^\bot_\rho\sigma\quad\textrm{ and }\quad K_d\,\beta^2(\rho,\sigma)-a\Tr P^\bot_\rho\sigma,
$$
respectively.\pagebreak

\textbf{Open question:} Is the constant $K_d$ in the second  inequalities  in (\ref{main+}) and in (\ref{main+++}) is optimal?
By Theorem 1 in \cite{H&Sh-5} this question can be reformulated as follows: \emph{does it exist a (global) supporting functional $\,\ell'_\rho(\sigma)=\Tr\Lambda'_{\rho}\sigma\,$ for the EoF  at the pure state $\rho$ with the Schmidt representation (\ref{smr}) generated by some Hermitian operator $\Lambda'_{\rho}$ with the spectral diameter less than $K_d$?} Note
that the spectral diameter of the operators  $\Lambda_{\rho}$  and $\widetilde{\Lambda}_{\rho}$  (defined in (\ref{old}) and in  (\ref{new})) is equal to $K_d$.\footnote{$K_{d}$ is the optimal constant in the log-Sobolev inequality for complete graph with $d$ vertices.}


\end{document}